\documentclass[runningheads,anonymous]{llncs}
\usepackage[T1]{fontenc}
\usepackage{graphicx}
\graphicspath{ {./images/} }
\usepackage{float}

\usepackage[namelimits]{amsmath} 
\usepackage{amssymb}
\usepackage{amsfonts}
\usepackage{algorithm}
\usepackage{algorithmicx}
\usepackage{algpseudocode}
\usepackage{amsmath} 
\usepackage{xurl}
\usepackage{multirow}
\usepackage{array}
\usepackage{booktabs}
\usepackage[table,xcdraw]{xcolor}

\usepackage{booktabs}
\usepackage{array}
\usepackage{graphicx}
\usepackage{caption} 
\usepackage{subcaption}

\usepackage{ulem}

\begin{document}
\title{ETGuard: Malicious Encrypted Traffic Detection in Blockchain-based Power Grid Systems }
\titlerunning{ETGuard}
\author{Peng Zhou\inst{1}\and
Yongdong Liu\inst{2, }\thanks{Yongdong Liu is the corresponding author. }\and
Lixun Ma\inst{2}\and
Weiye Zhang\inst{2}\and
Haohan Tan\inst{2}\and
Zhenguang Liu\inst{2}\and
Butian Huang\inst{2}}
\authorrunning{P. Zhou, Y. Liu, et al.}
\institute{State Grid Zhejiang Electric Power Company, LTD. Information and Communication Branch, China
\and
Zhejiang University, Hangzhou, Zhejiang Province, China\\
\email{\{malx\}@zju.edu.cn}}
\maketitle          % typeset the header of the contribution
\begin{abstract}
The escalating prevalence of encryption protocols has led to a concomitant surge in the number of malicious attacks that hide in encrypted traffic. Power grid systems, as fundamental infrastructure, are becoming prime targets for such attacks. Conventional methods for detecting malicious encrypted packets typically use a \textit{static} pre-trained model. We observe that these methods are not well-suited for blockchain-based power grid systems. More critically, they fall short in \textit{dynamic} environments where new types of encrypted attacks continuously emerge. 

\quad\; Motivated by this, in this paper we try to tackle these challenges from two aspects: 
(1) We present a novel framework that is able to automatically detect malicious encrypted traffic in blockchain-based power grid systems 
and incrementally learn from new malicious traffic. 
(2) We mathematically derive incremental learning losses to resist the forgetting of old attack patterns while ensuring the model is capable of handling new encrypted attack patterns. % 修改 Incremental Learning Objectives
Empirically, our method achieves state-of-the-art performance on three different benchmark datasets. 
We also constructed the first malicious encrypted traffic dataset for blockchain-based power grid scenario. 
Our code and dataset are available at https://github.
com/PPPmzt/ETGuard, hoping to inspire future research. 

\end{abstract}
\section{Introduction}
The \textit{immutable} and \textit{decentralized} nature of blockchain has led to applications such as Bitcoin, decentralized crowdfunding, and cross-industry finance \cite{zheng2018blockchain}. While research often focuses on system and software issues like \textit{consensus mechanisms} \cite{yadav2023comparative}, \textit{smart contracts} \cite{TIFSFuzzing}, and \textit{virtual machines}, cybersecurity challenges, particularly malicious traffic attacks, are often neglected.

In blockchain-based power systems, critical national infrastructure, malicious traffic poses severe risks such as \textit{widespread power outages} and \textit{energy data breaches}. Although encryption protocols are widely adopted to secure data, they can be exploited by attackers to hide malicious activities.
Detecting malicious encrypted traffic in power systems involves distinguishing attack patterns from benign packets \cite{li2021survey}. Research in this area generally follows two approaches: one \cite{gupta2020categorical,chiba2019newest,dong2021mbtree} uses decryption analysis to reveal clues in encrypted sequences, while another \cite{ni2023high,qing2023low} examines statistical differences using deep learning. Current methods face two main issues: poor performance on blockchain-based power grids, leading to a significant drop in F1 scores, and difficulties with novel attack sequences due to reliance on static models. 
%\cite{ni2023high,fu2023point,qing2023low}

To address these challenges, we propose a novel approach incorporating incremental learning to adapt to new attacks. Our method involves training a model with self-supervised learning to extract detailed packet features and deriving incremental learning losses to preserve old data while learning new attack patterns. We update the model using replayed samples combined with these losses. Additionally, we introduce the GridET-2024 dataset, which includes real-world traffic data from the State Grid of China, to evaluate detection in blockchain-based power grid scenarios.
To evaluate our method, we conducted extensive experiments on three benchmark datasets, as well as ablation studies to assess key components. Our method demonstrates state-of-the-art performance on these datasets.
In summary, our contributions are as follows:
\begin{itemize}
    \item[$ \bullet $] We propose ETGuard, a novel framework for detecting malicious encrypted traffic in blockchain-based power grids. It is the first method to automatically identify these attacks and adapt to new traffic patterns incrementally. 
    \item[$ \bullet $] We derive a loss function for effective incremental learning, supported by rigorous theoretical analysis, to manage new attack patterns while mitigating catastrophic forgetting. Additionally, we introduce a sample buffer for efficient storage and replay of representative traffic samples. 
    \item[$ \bullet $] We have collected real-world data from blockchain-based power grid systems and created the GridET-2024 dataset, the first of its kind for detecting malicious encrypted traffic in this context. Our method achieves state-of-the-art performance on several benchmark datasets. 
\end{itemize}

\section{Problem Statement}

\textbf{Problem formulation.} % 删除？
Given a sequence of encrypted packets \( s = \{p_1, p_2,
\ldots, p_n\} \), % 删除？
we are interested in developing a fully automated model to determine whether the packet sequence is malicious. 
Put differently, we aim to estimate the label \(\hat{y}\) for each encrypted packet sequence \( s \), 
where \(\hat{y} = 1\) represents \( s \) is a malicious sequence, 
and \(\hat{y} = 0\) indicates that \( s \) is benign.

\section{Method}
\textbf{Method Overview}\quad
The detailed architecture of our proposed framework is outlined in Fig.~\ref{Fig2}. 
Overall, the framework consists of three key components:
\begin{itemize}
    \item[$ \bullet $] \textit{Data Preprocessing}: Raw packets are cleaned and processed, ensuring that the packets from an individual client are sorted into a packet sequence and are separated from the packets of other clients. 
    To extract features from these sequences, we use an unsupervised auto-encoder with stacked bi-GRUs. 
    \item[$ \bullet $] \textit{Incremental Learning Module}: 
    To adapt to novel attacks while preventing catastrophic forgetting, we introduce an incremental learning module and mathematically derive the incremental learning losses. 
    
    \item[$ \bullet $] \textit{Detection Module}: The detection module learns the feature distinctions between benign and malicious sequences to continuously identify potential attacks. 
    The learning process is supervised by the classification loss and incremental learning losses. 
\end{itemize}
In what follows, we will elaborate on the details of these components one by one.

\begin{figure*}
    \vspace{-.5em}  
    \centering
    \includegraphics[width=1\textwidth]{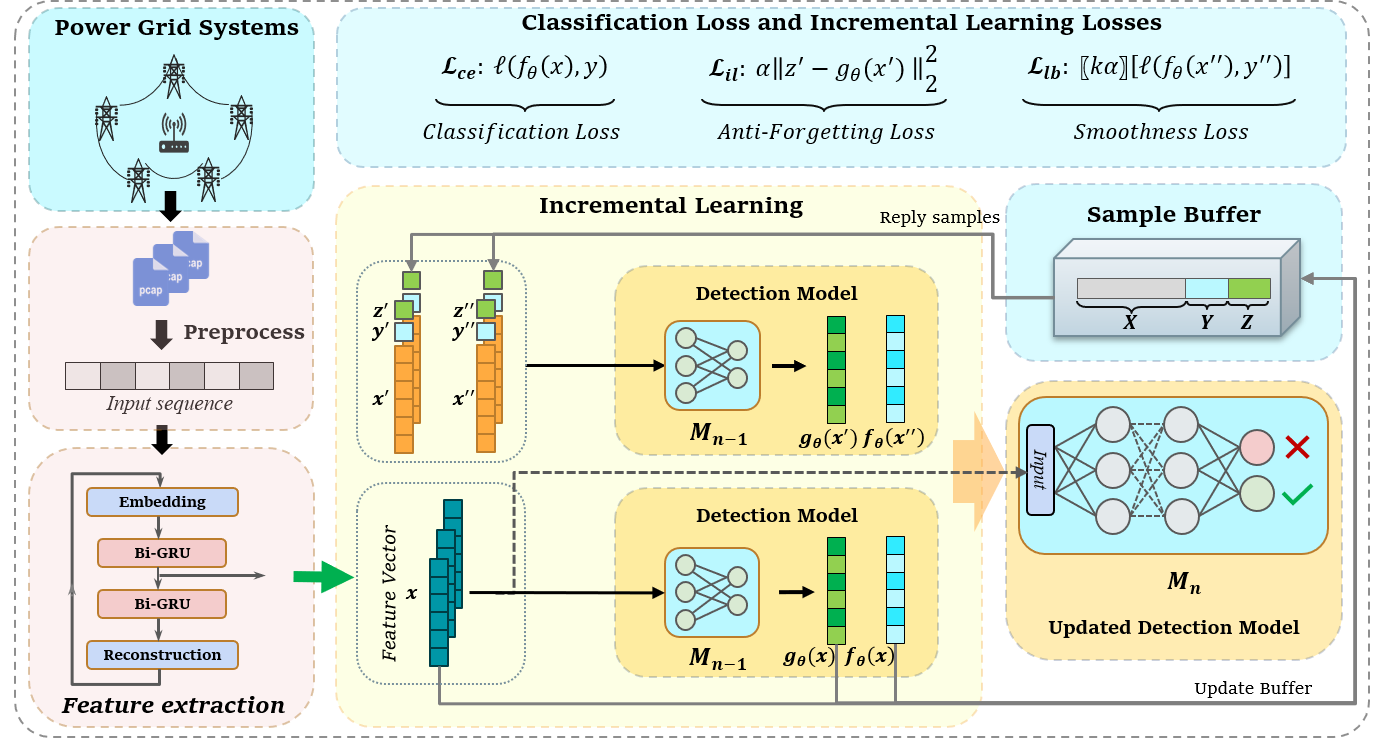}
    \caption{The framework of ETGuard.} 
    \label{Fig2}
    \vspace{-1.5em}
\end{figure*}
\subsection{Data Preprocessing}

The preprocessing module aims to clean and process raw packet data into distinct sequences for different clients. Since packet sequences cannot be directly input into a network, we employ an unsupervised auto-encoder with stacked bi-GRUs to extract features from each client's packet sequence.

Specifically, raw traffic data consists of packets organized by a five-tuple (\textit{i.e.}, source and destination IP addresses, source and destination ports, and transport layer protocol). We group packets with the same five-tuple into sequences $s$ and sort them chronologically. Irrelevant packet information, such as IP addresses and port numbers, is removed from each sequence.

The resulting sequences, $s_{input}$, are formatted with components $l$, $d$, and $t_m$, where: $l=\{b_1, b_2, \ldots, b_n\}$ denotes the packet length sequence, $d = t_n - t_1$ represents the duration of $s_{input}$, $t_m$ is the mean time interval between packets.

The feature extractor uses an auto-encoder with multiple bi-directional Gated Recurrent Units (bi-GRUs) \cite{chung2014empirical}, serving as both encoder and decoder. The encoder transforms the input sequence into a feature vector, which the decoder uses to reconstruct the sequence. A multi-layer perceptron in the reconstruction layer restores the original embedded sequence. The encoder learns an accurate representation of encrypted network traffic by minimizing reconstruction loss during training.

Compared to traditional methods \cite{anderson2016identifying,anderson2017machine,liu2018mampf} that focus on specific versions of TLS handshake metadata or message types \cite{qing2023low}, our approach captures the fine-grained behavior of network traffic more effectively.

\subsection{Mathematical Derivation of Incremental Learning Objectives}
In this subsection, we introduce the incremental learning module, and provide the key mathematical derivations of incremental learning objectives within this module.

We use the incremental learning method based on empirical replay \cite{rebuffi2017icarl,buzzega2020dark,van2020brain} with targeted modifications for encrypted traffic scenarios 
to realize the incremental update of the model when facing new types of traffic attack.
Specifically, we maintain a sample buffer to store representative traffic samples.
When new malicious encrypted traffic arises, we incrementally update the model through a anti-forgetting loss function to detect new attack patterns. % 修改
The new traffic samples are updated into the buffer using a reservoir sampling algorithm \cite{vitter1985random}, serving as a representative sample set for subsequent incremental model updates. %, as shown in Fig.~\ref{Fig1}.

\subsubsection{Sample Buffer}
The sample buffer is used to store representative traffic samples to achieve experience replay during incremental learning.
When the model learns new traffic patterns, it replays previous samples from the sample buffer to effectively prevent catastrophic forgetting. 
We implement the update of the sample buffer using the reservoir sampling algorithm.
When new traffic samples arrive, the sample buffer is dynamically updated to ensure diversity and representativeness.

\subsubsection{Incremental Learning Loss Function}
The goal of the incremental learning module is to detect the ongoing emergence of new malicious traffic attacks. To enable the model to learn from new datasets, the loss function and optimization objectives are defined as follows:

\begin{equation}
    \min_{\theta}\mathcal L_{ce}=\mathbb{E}_{(x,y)\sim D}[\ell(f_{\theta}(x),y)], 
\end{equation}
where $\theta$ denotes the model parameters, $\mathcal L_{ce}$ represents the loss function using cross-entropy, 
$D$ denotes the dataset, $f_\theta(x)$ denotes the detection model, and $y$ denotes the true target of the sample.

To mitigate catastrophic forgetting, we introduce a new loss function aimed at balancing the learning between new and old data. We use the past data on the updated model to obtain an output $f_\theta(x)$ that closely approximates the output of the pre-update model trained on past traffic data, denoted as $f_{\theta*t}$.

\begin{equation}
    \mathcal L_{il}=\alpha \mathbb{E}_{(x,z)\sim \mathcal{M}} \left[ D_{KL} (\text{softmax}(z) || f_{\theta}(x)) \right].
\end{equation}

In the above equation, $z$ represents the output logits of the model on the past traffic data. 
In order to reduce the resource consumption, we use $z$ to replace the model output $f_{\theta*t}$. 
The logits $z$ are softmaxed to obtain the probability vector, which is used to calculate the KL scatter with the model output $f_\theta(x)$.

\begin{theorem}
\textit{If two logits output by the model are similar, the KL divergence between them can be approximated as the Euclidean distance.}
\label{theorem1}
\end{theorem}

\begin{proof}
Assume logits \( z_1 \) and logits \( z_2 \) are two close vectors, we set the sample space of logits $z$, $R=\{1,...,n\}$, $n$ is the dimension of logits $z$, so we can express \( z_2 \) as:
\begin{equation}
    z_2 = z_1 + \epsilon \Delta z, 
\end{equation}
for some small \( \epsilon > 0 \) and perturbation vector \( \Delta z = [\Delta z(1) \cdots \Delta z(n)]^T \). 
Next, we approximate the KL divergence locally.
\begin{equation}
    D_{KL}(z_1 \| z_2) = \sum_{i \in R} z_1(i) \log \left( \frac{z_1(i)}{z_2(i)} \right). 
\end{equation}
Substituting \( z_2(i) = z_1(i) + \epsilon \Delta z(i) \) into the KL divergence formula gives:
\begin{equation}
    D_{KL}(z_1 \| z_2) = \sum_{i \in R} z_1(i) \log \left( \frac{z_1(i)}{z_1(i) + \epsilon \Delta z(i)} \right). 
\end{equation}
Applying the second-order Taylor approximation to \( \log \left( \frac{z_1(i)}{z_1(i) + \epsilon \Delta z(i)} \right) \):
\begin{equation}
    D_{KL}(z_1 \| z_2) = \frac{\epsilon^2}{2} \sum_{i \in R} \frac{\Delta z(i)^2}{z_1(i)} + o(\epsilon^2). 
\end{equation}
This implicitly assumes \( z_1(i) > 0 \) for all \( i \in R \). 
Introducing the weighted Euclidean norm, the KL divergence becomes:
\begin{equation}
    D_{KL}(z_1 \| z_2) = \frac{\epsilon^2}{2} \| \Delta z \|_{z_1}^2 + o(\epsilon^2)=\|z_1 - z_2\|. 
\end{equation}
Thus, under the local approximation, the KL divergence approximates a weighted Euclidean distance.
\end{proof}

By Theorem~\ref{theorem1}, $\mathcal{L}_{il}$ can be simplified to: 
\begin{equation}
    \mathcal{L}_{il} = \alpha \mathbb{E}_{(x,z)\sim \mathcal{M}} \left[ \left\| z' - g(x') \right\|_2^2 \right]. 
\end{equation}

To handle significant changes in new attack patterns and avoid bias towards previous patterns, we introduce a smoothness loss function:
\begin{equation}
    \mathcal{L}_{lb} = \beta \mathbb{E}_{(x'', y'', z'') \sim \mathcal{M}} \left[ \ell(y'', h(x'')) \right].
\end{equation}

To balance the contributions of $\mathcal{L}_{il}$ and $\mathcal{L}_{lb}$, we use a coefficient $k$:
\begin{equation}
    k = 0.5 + \text{softmax}(\mathcal{L}_{il} \cdot \gamma).
\end{equation}

The final loss function is:
\begin{equation}
    \mathcal{L}_{ce} + \alpha \mathbb{E}_{(x', y', z') \sim \mathcal{M}} \left[ \left\| z' - g(x') \right\|_2^2 \right] + k \alpha \mathbb{E}_{(x'', y'', z'') \sim \mathcal{M}} \left[ \ell(y'', h(x'')) \right].
\end{equation}

\subsection{Detection Module}
Due to the substantial volume of traffic data and the high traffic rate in the blockchain-based power grid scenario, the real-time performance and resource consumption of the model are critically demanding. The MLP model architecture, being relatively simple, requires lower computing resources and offers faster training speeds. It is capable of monitoring traffic data in real time, and experiments have demonstrated that the MLP model is sufficient to meet the task requirements. Therefore, the MLP model is chosen to detect malicious traffic.

\section{Evaluations}
In this section, we conduct extensive experiments on multiple malicious encrypted traffic detection datasets to evaluate our framework. 
Next, we introduce the experimental setup, followed by presenting the comprehensive empirical results.
\subsection{Experimental Setup}
\subsubsection{Datasets}
\begin{itemize}
    \item[$ \bullet $] \textbf{CIRA-CIC-DoHBrw-2020 (DoHBrw) \cite{montazerishatoori2020detection}:}
    The DoHBrw dataset provides a mix of benign and malicious DNS-over-HTTPS (DoH) traffic, all data is encrypted traffic.
    The normal traffic is generated by querying benign DNS servers using the DoH  protocol.
    Tunneling tools such as dns2tcp, DNSCat2, and Iodine are used to generate malicious DoH traffic.
    \item[$ \bullet $] \textbf{CIC-AndMal2017 (CIC) \cite{lashkari2018toward}:}
    CIC collected a rich variety of malicious attacks from several sources.
    The malicious traffic samples come from 42 unique malware families, 
    which can be classified into four categories: Adware, Ransomware, Scareware, and SMS Malware.
    \item[$ \bullet $] \textbf{GridET-2024 (GridET):}
    To better detect the encrypted attacks of real-world blockchain-based power grid scenario, we create the dataset GridET-2024.
    Benign traffic samples are collected by capturing power grid system interaction traffic data. 
    Malicious traffic data samples are sourced from malware-traffic-analysis.net and USTC-TFC2016 dataset to ensure the diversity of malicious traffic attack patterns.
\end{itemize}

\subsubsection{Implementation Details}
We implement our detection framework and all baselines by using Python 3.8.5.
% The used libraries include NumPy 1.21.2, Pytorch 1.9.1, Tensorflow 2.10.0, and Cuda 11.7. 
We run these models on a Linux server with NVIDIA GeForce RTX 3090 GPU.
We list all the parameters used by our framework in Table \ref{tab2}. 
In particular, we set $n$ = 50 and $d$ = 32 to make a better capability of 
capturing fine-grained behaviors of network traffic. 
Then, we apply Grid Search to find the appropriate values for $\alpha$ and $\gamma$. 
The values of $\alpha$ and $\gamma$ are 0.5 and 10, respectively. 
In particular, buffer size is a hyperparameter manually tuned to best fit the specific scenario.
We conduct comprehensive experiments to evaluate the performance of ETGuard with various buffer sizes.

\begin{table}[h]
    \centering
    \begin{minipage}{0.40\textwidth}
        \centering
        \vspace{-3.5em}
        \caption{Statistics of Datasets}
        \begin{tabular}{l|cc}
            \toprule
            Dataset & Normal & Malicious \\
            \midrule
            DoHBrw & 688,489 & 6,112 \\
            CIC & 894,367 & 62,972 \\
            GridET & 43,611 & 27,141 \\
            \bottomrule
        \end{tabular}
        \label{tab1}
        \vspace{-3.5em}
    \end{minipage}%
    \hfill
    \begin{minipage}{0.6\textwidth} % Increased size
        \centering
        \vspace{-2em}
        \caption{Parameter Settings of ETGuard}
        \resizebox{\columnwidth}{!}{
        \begin{tabular}{lcccc}
            \toprule
            Module & Para. & Value & Description \\
            \midrule
            \multirow{4}{*}{Feature Extraction} 
            & $n$ & 50 & Number of used head packets \\
            & $V$ & 32 & Embedding size of GRU-AE \\
            & $H$ & 8 & Hidden size of each GRU layer \\
            & $B$ & 2 & Number of GRU layers \\
            \midrule
            \multirow{2}{*}{Incremental Learning} 
            & $\alpha$ & 0.5 & Coefficient of loss $L_2 $ \\
            & $\gamma$ & 10 & Coefficient to balance loss $L_2$ and loss $L_3$\\ 
            \bottomrule
        \end{tabular}
        }
        \label{tab2}
        \vspace{-2em}
    \end{minipage}
\end{table}

\subsubsection{Evaluation Metrics}
We use \textit{ACC} and \textit{F1 Score} as metrics to evaluate the malicious traffic detection and incremental learning performance of ETGuard. 

\subsection{Performance on Malicious Encrypted Traffic Detection}
In this section, we benchmark our method against state-of-the-art malicious encrypted traffic detection methods 
for two public dataset and one power grad scenario dataset GridET.

In public dataset evaluations, we train and test methods on DoHBrw, and CIC, respectively. 
Fig.~\ref{detec1} presents public dataset comparison results. 
From Fig.~\ref{detec1}, we observe that our method is capable of consistently outperforming existing methods on all five benchmarks. 
For example, the F1 score of our method is 0.92 on DoHBrw while the state-of-the-art detection method RAPIER \cite{qing2023low} is 0.88. 
The F1 score of our method is also outperformed FS \cite{liu2019fs} in all datasets.
In addition, we also use the CoinFlip algorithm and PacketLen algorithm, which simulate randomly guess and only utilize packet length, respectively, to detect encrypted traffic.

\begin{table}[htbp]
    \centering
    \vspace{-.5cm}
    \caption{F1 scores of Malicious Encrypted Traffic Detection Methods}
    \begin{tabular}{lccccc>{\centering\arraybackslash}p{5cm}}
        \toprule
        \textbf{Dataset} & \textbf{RAPIER} & \textbf{FS} & \textbf{CoinFlip} & \textbf{PacketLen} & \textbf{ETGuard (Ours)} \\ \midrule
        DoHBrw           & 0.88            & 0.76        & 0.28              & 0.56               & \textbf{0.92}           \\
        CIC              & 0.84            & 0.71        & 0.27              & 0.46               & \textbf{0.86}           \\
        GridET           & 0.83            & 0.73        & 0.31              & 0.49               & \textbf{0.94}           \\ \bottomrule
    \end{tabular}
    \label{111}
    \vspace{-0.5cm}
\end{table}

The blockchain-based power grid scenario malicious traffic detection is more challenging for existed detection methods. 
To evaluate the detection abilities of the methods on this scenario, we train and test the models on the GridET dataset. 
Table \ref{111} demonstrate the state-of-the-art malicious encrypted traffic detection methods still suffer from relatively low F1 score on the GridET dataset, 
which reveals that such methods are fall short to extract the critical features of traffic sample in the blockchain-based power grid scenario.

\begin{figure*}[htbp]%[H]
    \vspace{-1em} 
    \centering
    \includegraphics[width=0.7\textwidth]{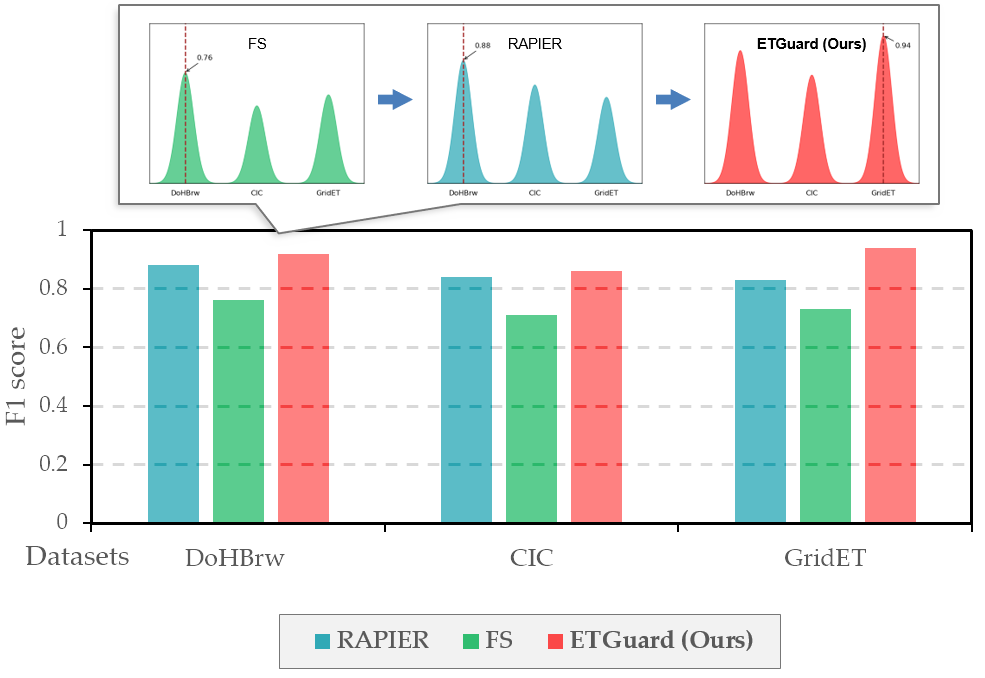}
    \caption{Performance on Malicious Encrypted Traffic Detection. } 
    \label{detec1}
    \vspace{-1.5em}
\end{figure*}

Overall, our method achieves state-of-the-art general scenarios and blockchain-based power grid scenario malicious encrypted traffic detection performance. 
For general scenarios comparisons, our method improves the F1 score on DoHBrw from 0.88 to 0.92, and on CIC from 0.84 to 0.86.
In contrast with general scenarios methods, our method also attains 0.94 F1 score on GridET, outperforming the current state-of-the-art method RAPIER. 

\subsection{Performance on Incremental Learning}
To evaluate the performance of our method on incremental learning, we create a new dataset \textit{DoHBrw/CIC}.
Specifically, we combine all benign samples from the DoHBrw dataset with a selection of malicious samples from the CIC dataset.
We further divide the datasets into six sub-datasets $\{A_0,A_1,A_2,A_3,A_4,A_5\}$.
In each of these sub-datasets, the benign traffic is all of the same type DoHBrw, 
while the malicious traffic all consists of different types of malicious attacks. %, as shown in Table \ref{dataset2}.
We use $A_0$ for pre-training the model, while the other sub-datasets are used for incremental updates to the model. 
We established a separate test set for each round, where the test set for round i includes all attack types observed from rounds 0 to i. 

We compare ETGuard against five incremental learning methods (ER \cite{riemer2018learning}, DER \cite{buzzega2020dark}, DER++ \cite{buzzega2020dark}, GSS \cite{aljundi2019gradient}, SI \cite{zenke2017continual}) on DoHBrw/CIC. 
To assess the efficacy of the incremental learning module within our approach, we extracted this component from ETGuard, resulting in a variant dubbed ETGuard-V.
We then evaluated the performance of ETGuard-V to conduct an ablation study.
We further provide an upper bound given by training all attack samples (FULL).

\begin{figure}[H]%[htbp]
    \centering
    \vspace{-1cm}
    \includegraphics[width=.623\linewidth]{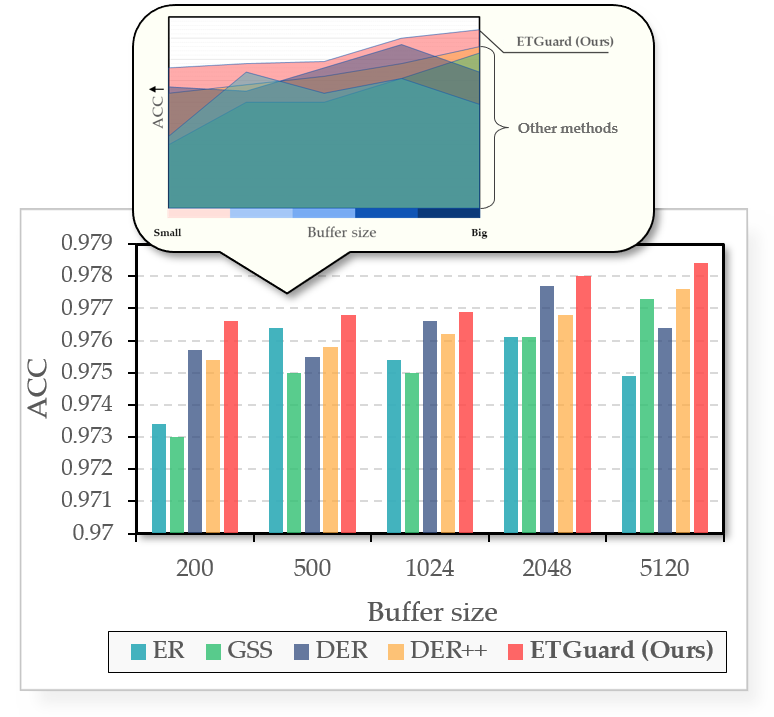}
    \caption{The Performance of Incremental Learning Methods.}
    \label{incre1}
    \vspace{-0.21cm}
\end{figure}

Fig.~\ref{incre1} reports performance in terms of average accuracy across all rounds. %(results are averaged across ten runs).
Experimental evidence ETGuard achieve state-of-the-art performance in almost all settings.

\begin{figure}[H]%[htbp]
    \centering
    \vspace{-.5cm}
    \includegraphics[width=.7\linewidth]{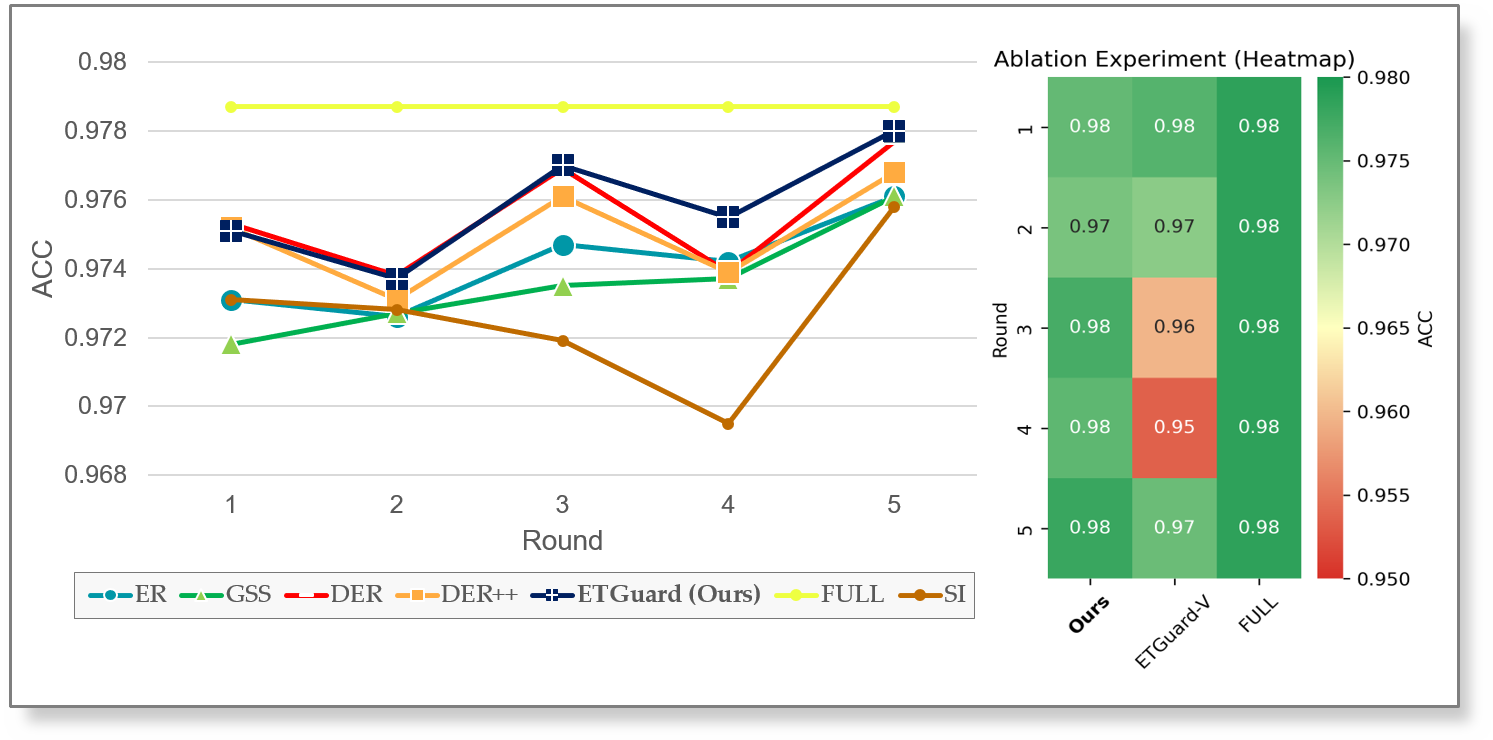}
    \caption{The Performance of Incremental Learning Methods in Different Round.}
    \label{incre2}
    \vspace{-.8cm}
\end{figure}

At the same time, we observe that the performance of ETGuard is almost always better than that of ETGuard-V. 
And as the number of rounds increases, the gap in detection performance between ETGuard-V and ETGuard gradually widens, 
which further proves the effectiveness of our incremental learning module.

\section{Related Work}
\subsection{Malicious Encrypted Traffic Detection}
Traditional malicious encrypted traffic detection mainly uses signature-based methods \cite{gupta2020categorical,chiba2019newest,dong2021mbtree} to detect malicious encrypted traffic. However, the method relies heavily on the quality of decryption operations and rules for traffic \cite{khraisat2019survey,azeez2020intrusion}. 
With the growth of artificial intelligence technology\cite{ACMMMdeepfake,IJCAIvideoGeneration}, machine learning is increasingly being adopted for detecting malicious encrypted traffic. 
Machine learning enhances detection by extracting statistical features from traffic, offering faster and more accurate results compared to traditional methods. 
For example, Fu et al, utilized frequency domain features for real-time detection \cite{fu2021realtime}. Barradas et al, detect attacks by applying random forests \cite{barradas2021flowlens}.
In addition to traditional traffic detection or packet inspection \cite{mirsky2018kitsune,anderson2017machine}, 
Fang et al. \cite{fang2021communication}, on the other hand, detects TLS traffic by collecting features of the traffic communication channel (packets consisting of the same destination IP and destination port) 
and uses Random Forest (RF) to enhance malware traffic detection performance.
All these methods are not effective in detecting attacks on new encrypted traffic. 

\subsection{Incremental Learning}
The core challenge of incremental learning is to balance the conflict between remembering information about old tasks and absorbing information about new tasks, the so-called catastrophic forgetting problem. To overcome this problem, existing methods fall into two main categories: replay-based methods and parameter optimization-based methods.
Replay-based methods mitigate Catastrophic Forgetting by replaying some samples of old tasks while learning new ones. The replayed samples can be real historical data, \textit{i.e.}, empirical replay. It can also be pseudo-samples generated by generative models (\textit{e.g.}, Generative Adversarial Networks, GAN), \textit{i.e.}, generative replay. iCaRL \cite{rebuffi2017icarl} is a representative of the empirical replay-based approach, which combines knowledge distillation methods to update the model parameters on a representative sample pool. However, iCaRL updates the parameters of old tasks and therefore suffers from overfitting to old data. 

\section{Conclusion}
In this paper, we try to tackle the malicious encrypted traffic detection problem from two aspects: 
(1) We propose a novel framework termed ETGuard, which to our knowledge is the first approach tailored for automatically identifying malicious traffic attacks in blockchain-based power grid systems. 
(2) We lay the mathematical foundation for establishing an incremental learning model that can 
effectively adapt to new types of attacks. 
We utilized real data collected from the State Grid and constructed the malicious encrypted traffic dataset GridET.
We extensively evaluated the proposed method on three benchmark datasets. 
Empirical results show that our method consistently delivers state-of-the-art performance on malicious encrypted traffic detection across general scenarios and the blockchain-based power grid scenario. 

\begin{credits}
\subsubsection{\ackname} This work was supported by State Grid Zhejiang Electric Power Company, LTD. Information and Communication Branch, China (Grant number 5211XT
24000D). 
\end{credits}

% \clearpage
% \newpage

\bibliography{reference}
\end{document}